%% file: mainDualGainScheduling.tex
\documentclass[letterpaper, 10 pt, conference]{ieeeconf}  

\IEEEoverridecommandlockouts                              

\usepackage{graphics} 
\usepackage{epsfig} 
\usepackage{amsmath} 
\usepackage{amssymb}  
\usepackage{cite}

\usepackage{epstopdf}

\usepackage{verbatim}

\usepackage{amsthm}

\usepackage{color}
\usepackage{cancel}

\usepackage{tikz}               
\usepgflibrary{arrows}
\usepackage{float}

\usepackage{algpseudocode}
\usepackage{algorithm}
\usepackage{verbatim}

\usepackage{setspace}

\newtheorem{assumption}{\bf Assumption}

\newtheorem{theorem}{\bf Theorem}
\newtheorem{proposition}{\bf Proposition}
\newtheorem{lemma}{\bf Lemma}

\DeclareMathOperator*{\argmin}{arg\,min}

\title{Robust Dual Control based on Gain Scheduling}
\author{Janani Venkatasubramanian, Johannes K\"ohler, Julian Berberich, Frank Allg\"ower
\thanks{The authors are with the Institute for Systems Theory and Automatic Control, University of Stuttgart,70550 Stuttgart, Germany. (email:$\{$janani.venkatasubramanian, johannes.koehler, julian.berberich, frank.allgower\}@ist.uni-stuttgart.de)}%
\thanks{This work was funded by Deutsche Forschungsgemeinschaft (DFG,
German Research Foundation) under Germany’s Excellence Strategy - EXC
2075 - 390740016. The authors thank the International Max Planck Research
School for Intelligent Systems (IMPRS-IS) for supporting Janani Venkatasubramanian and Julian Berberich,
and the International Research Training Group Soft Tissue Robotics (GRK
2198/1).}
}
\allowdisplaybreaks

\begin{document}
\maketitle
\thispagestyle{empty}
\pagestyle{empty}
\begin{abstract}
\input{abstract}
\end{abstract}

\input{Introduction}
\input{setup}
\input{preliminaries}

\input{dualcontrol}
\input{conclusion}
\bibliographystyle{IEEEtran}  
\bibliography{lit}

\end{document}

%% file: abstract.tex
We present a novel strategy for robust dual control of linear time-invariant systems based on gain scheduling with performance guarantees. This work relies on prior results of determining uncertainty bounds of system parameters estimated through exploration. Existing approaches are unable to account for changes of the mean of system parameters in the exploration phase and thus to accurately capture the \textit{dual} effect. We address this limitation by selecting the future (uncertain) mean as a scheduling variable in the control design. 
The result is a  semi-definite program-based design that computes a suitable exploration strategy and a robust gain-scheduled controller with probabilistic quadratic performance bounds after the exploration phase.

%% file: Introduction.tex
\section{Introduction}
The \textit{dual control} paradigm established research interest in simultaneous learning and control of uncertain dynamic systems\cite{feldbaum1960dual}. This pioneering work recognized that control inputs to an uncertain system have a `probing' effect to learn the uncertainty in the system, and a `directing' effect to control the dynamical system.  However, these two effects are naturally conflicting, drawing attention to the trade-off between `exploration' (learning system uncertainty) and `exploitation' (controlling the system to achieve optimal performance), which is also the subject of contemporary literature on Reinforcement Learning \cite{recht2019tour}. Dual control relies on stochastic Dynamic Programming (DP) which is, however, computationally intractable. Either approximations of stochastic DP, or heuristic probing methods are typically adopted to solve the problem of tractability \cite{filatov2000survey}. A detailed survey of dual control methods is provided in \cite{mesbah2018stochastic}.

Early works of implicit dual control methods that involve approximations of DP are based on the \textit{wide-sense} property \cite{doob1953stochastic}, but they require linearization of system dynamics and approximation of the conditional probability of the states by its mean and covariance \cite{tse1973wide, tse1976actively}. These methods were extended to nonlinear systems that could handle input constraints, nonetheless based on some approximations \cite{bayard1985implicit}. This laid the foundation of balancing exploration with caution \cite{bar1976caution}. 

Explicit dual control methods use heuristic probing techniques for active learning without the need to introduce approximations of DP \cite{wittenmark1995adaptive}. Explicit dual control methods are closely related to \textit{Optimal Experiment Design} in closed loop \cite{geversaa2005identification, hjalmarsson2005experiment}. These methods utilize the control inputs to regulate system dynamics and to probe the closed-loop system dynamics by solving a combined problem. This led to application-oriented strategies for dual control that promoted reducing uncertainty that would be beneficial for optimizing cost \cite{annergren2017application}. Some recent application-oriented strategies are discussed in \cite{larsson2016application, heirung2017dual}, however, they consider a special class of systems and their control strategies are not robust to model uncertainties. The `coarse-ID' family of methods study robustness guarantees in system identification based design methods \cite{dean2018regret, dean2019safely, dean2017sample}, however, the control policies are not optimized to balance exploration and exploitation.

Recently, in \cite{umenberger2019robust} a high probability uncertainty bound has been derived that is applicable to both robust control synthesis and \textit{targeted} exploration. This bound is then used in dual control by predicting the influence of the controller on the future uncertainty. In particular, the work in \cite{ferizbegovic2019learning}, building on \cite{barenthin2008identification, umenberger2019robust}, proposes a dual control strategy that minimizes worst-case cost attained by a robust controller that is synthesized with reduced model uncertainty. This dual control strategy with \textit{targeted} exploration seems to perform better than strategies with common greedy random exploration. The approach in \cite{iannelli2019structured}, further extending \cite{umenberger2019robust}, retains an application-oriented strategy, however, adopts a more realistic finite horizon problem setting that captures the trade-offs between exploration and exploitation better. 

The results presented in the methods in \cite{umenberger2019robust, ferizbegovic2019learning, barenthin2008identification, iannelli2019structured} seem to have lower conservatism compared to previous works  \cite{geversaa2005identification, hjalmarsson2005experiment, annergren2017application, larsson2016application}, however, only numerically without any performance guarantees. In particular, the approaches in \cite{umenberger2019robust, ferizbegovic2019learning} do not account for changes in the mean of uncertain system parameters during exploration. Therefore, this paper seeks to address these drawbacks by designing a dual control scheme based on gain scheduling. The mean of future uncertain system parameters is selected as a scheduling variable. This leads to a linear matrix inequality (LMI) based design with guarantees under relaxed assumptions.
In particular, the resulting controller is a state feedback, which depends on the parameter estimates after exploration and thus on the data, and it \emph{guarantees} robust closed-loop performance after an initial exploration phase.

The remainder of the paper is structured as follows.
In Section~\ref{sec:problem} we state the problem setting, and in Section~\ref{sec:prelim} we provide important results from the literature that we employ for our approach.
Section~\ref{sec:dual} contains the proposed dual controller design procedure as well as a proof of robust closed-loop performance guarantees.
Finally, we conclude the paper in Section~\ref{sec:conclusion}.

%% file: setup.tex
\section{Problem Statement}\label{sec:problem}
\subsubsection*{Notation} 
The transpose of a matrix $A$ is denoted as $A^\top$. 
The value of the Chi-squared distribution with $n$ degrees of freedom and probability $p$ is denoted as $\chi_n^2(p)$. $\mathcal{L}_2$ denotes the space of square-summable functions.

\subsubsection*{Setting}
Consider a discrete-time linear time-invariant dynamical system of the form \begin{equation}\label{sys}
x_{k+1}=A_{tr}x_k+B_{tr}u_k+w_k,\quad w_k \sim \mathcal{N}(0,\sigma_w^2I)
\end{equation}
where $x_k \in \mathbb{R}^{n_x}$ is the state, $u_k \in \mathbb{R}^{n_u}$ is the control input, and $w_k \in \mathbb{R}^{n_x}$ is the normally distributed process noise. The true values of the system dynamics, $A_{tr}$ and $B_{tr}$, are unknown.
\subsubsection*{Control goal - proposed approach}
The main goal is to design a stabilizing state-feedback $u_k=K_{\text{new}}x_k$ which meets some desired performance specifications. Since the system is unknown, we first apply some exciting input to estimate the parameters and then use bounds on the estimation error to design a robust controller. Our goal is to simultaneously design a suitable exploration strategy and a controller for the system in (\ref{sys}), such that applying the feedback after the exploration phase provides desired quadratic performance guarantees with high probability. The main challenge is to accurately capture the \textit{dual} effect of performance improvement through exploration. We solve this problem by interpreting the new parameter estimate as a scheduling variable, which influences the control law $K_{\text{new}}$, using tools from gain-scheduling. 
The corresponding necessary preliminaries regarding uncertainty bounds for parameter estimation, gain-scheduling and structured exploration are show in Section~\ref{sec:prelim} and the overall approach is presented in Section~\ref{sec:dual}.

%% file: preliminaries.tex
\section{Preliminaries}\label{sec:prelim}
\subsection{Uncertainty Bound}
This subsection discusses preliminary results from \cite{umenberger2019robust} adopted in our work that quantify uncertainty in the system dynamics that are estimated, given some data.
The unknown matrices $A_{tr}$ and $B_{tr}$ can be estimated through observed data $\mathcal{D}=\{x_k,u_k\}_{k=0}^N$ of length $N$. In particular, we consider the least squares estimates of $A_{tr}$ and $B_{tr}$, similar to \cite{umenberger2019robust}, which are given by,
\begin{equation}
\label{eq:LMS}
(\hat{A},\hat{B})= \argmin_{A,B} \sum_{k=0}^{N-1} ||x_{k+1}-Ax_k - Bu_k||_2^2.
\end{equation}
%
The following lemma provides a high probability credibility region for the uncertain system matrices.
\begin{lemma}\label{lem1} \cite[Lemma 3.1]{umenberger2019robust}
Given data set $\mathcal{D}$ and $0 < \delta <1$, let 
$D= \frac{1}{\sigma_w^2 c_\delta} \sum_{k=1}^{N-1}\begin{bmatrix}x_k\\u_k\end{bmatrix}\begin{bmatrix}x_k\\u_k\end{bmatrix}^\top$ 
with $c_\delta=\chi_{n_x^2+n_xn_u}^{2}(\delta)$. Suppose we have a uniform prior over the parameters $(A,B)$.
Then, $[A_{tr},B_{tr}] \in \Theta$ with probability $1-\delta$, where
\begin{equation}\label{credregion}
\Theta:=\Bigg\{A,B:\begin{bmatrix}
(\hat{A}-A)^\top\\(\hat{B}-B)^\top
\end{bmatrix}^\top D \begin{bmatrix}
(\hat{A}-A)^\top\\(\hat{B}-B)^\top
\end{bmatrix} \preceq I\Bigg\}.
\end{equation}

\end{lemma}
This lemma provides a data-dependent uncertainty bound. Given this uncertainty bound, the approaches in \cite{umenberger2019robust, ferizbegovic2019learning, iannelli2019structured} synthesize a robust controller by minimizing a worst-case cost. This controller facilitates \textit{targeted} exploration for dual control strategies by predicting the future uncertainty, depending on the exploring controller. However, their approach does not take into account that the estimate of the system parameters are subject to change through the process of exploration.

\subsection{Gain Scheduling Approach}
To account for the change in the estimates of the system parameters, we model the system in (\ref{sys}) as a linear parameter varying (LPV) system. The varying system parameters can be measured after exploration and are selected as the \textit{scheduling block}. The goal is to design a gain-scheduling controller that ensures that the closed-loop system is stable while also satisfying a quadratic performance bound, e.g. $\mathcal{L}_2$ gain, from the disturbance input $w$ to the performance output $z$ with high probability, compare \cite{scherer2001lpv, veenman2014synthesis}. The performance specification is imposed on the channel $w \rightarrow z$, where the performance output $z_k$ at time $k$ is the generalized error that depends on the state, control input and disturbance:
\begin{equation}\label{gz}
z_k=Cx_k + Du_k +D_ww_k,
\end{equation}
where $C$, $D$ and $D_w$ are known. In this setup, since the dynamics are unknown, we have the following assumption from which an initial error bound of the form given in Lemma \ref{lem1} can be derived.
\begin{assumption} \label{a1}
An initial data set $\mathcal{D}_0=\{x_t,u_t\}_{k=-N_0}^{-1}$ is available and a uniform prior over the parameters $(A,B)$ is assumed.
Moreover, it holds that
\begin{align}\label{eq:D0}
D_0:=\frac{1}{\sigma_w^2c_\delta}\sum_{k=-N_0}^{-1}\begin{bmatrix}x_k\\u_k\end{bmatrix}\begin{bmatrix}x_k\\u_k\end{bmatrix}^\top\succ0.
\end{align}
\end{assumption}
From the data $\mathcal{D}_0$, initial estimates of the system parameters $\hat{A}_0$ and $\hat{B}_0$ can be derived. The matrix $D_0$ quantifies the uncertainty associated with these initial estimates for a given probability $1-\delta$ and can be determined from $\mathcal{D}_0$ as given in~\eqref{eq:D0}.
 This initial data can be acquired through some random persistently exciting input, while the later exploration will use the initially obtained model knowledge to provide a more \textit{targeted} exploration strategy. Through the exploration process for $T$ time steps, data $\mathcal{D}_T=\{x_t,u_t\}_{t=0}^{T}$ will be observed. The new estimates $\hat{A}_T$ and $\hat{B}_T$ will be computed from data $\mathcal{D}_0 \cup \mathcal{D}_T$ and made available at time $T$. 
 The matrix $D_T:=D_0+\frac{1}{\sigma_w^2c_\delta}\sum_{k=0}^{T-1}\begin{bmatrix}x_k\\u_k\end{bmatrix}\begin{bmatrix}x_k\\u_k\end{bmatrix}^\top$ will quantify the uncertainty associated with the estimates $\hat{A}_T$ and $\hat{B}_T$. Existing approaches such as \cite{ferizbegovic2019learning} rely on the assumption that $\hat{A}_0\approx\hat{A}_T$ and $\hat{B}_0\approx\hat{B}_T$. In the following, we propose a gain scheduling-based approach to provide closed-loop guarantees for the case $\hat{A}_0\neq \hat{A}_T$ and $\hat{B}_0\neq\hat{B}_T$.
Since the system parameters will be updated through the process of exploration, we proceed now by rewriting (\ref{sys}) as,
\begin{align}\label{gx}
x_{k+1}=&A_{tr}x_k+B_{tr}u_k+w_k\\\nonumber
=&\hat{A}_0x_k+\hat{B}_0u_k+(\hat{A}_T-\hat{A}_0)x_k+(\hat{B}_T-\hat{B}_0)u_k\\\nonumber
& + (A_{tr}-\hat{A}_T)x_k +(B_{tr}-\hat{B}_T)u_k +w_k.
\end{align}

From (\ref{gx}), the scheduling and uncertainty blocks can be selected as,
\begin{equation}
\begin{split}
\Delta_s & = \begin{bmatrix}
\hat{A}_T-\hat{A}_0 & \hat{B}_T-\hat{B}_0
\end{bmatrix}, \\
\Delta_u & = \begin{bmatrix}
A_{tr}-\hat{A}_T & B_{tr}-\hat{B}_T
\end{bmatrix}.
\end{split}
\end{equation}

Since the estimates at time $T$ affect both $\Delta_s$ and $\Delta_u$, the latter blocks can be viewed as time-varying parameters, and the uncertain system combining (\ref{gz}) and (\ref{gx}) can be written as an LPV system:

\begin{equation}\label{newlpv}
\begin{gathered}
\begin{bmatrix}
x_{k+1}\\z_k^s \\ z_k^u \\z_k
\end{bmatrix}   = 
\begin{bmatrix}
\hat{A}_0 & I & I & I & \hat{B}_0\\
\begin{bmatrix}
I\\0
\end{bmatrix} & 0 & 0 & 0 & \begin{bmatrix}
0 \\ I
\end{bmatrix}\\
\begin{bmatrix}
I\\0
\end{bmatrix} & 0 & 0 & 0 & \begin{bmatrix}
0 \\ I
\end{bmatrix}\\
C & 0 & 0 & D_w & D
\end{bmatrix}
\begin{bmatrix}
x_k\\w_k^s\\w_k^u\\w_k\\u_k
\end{bmatrix},\\
w_k^s  = \Delta_s z_k^s,\\
w_k^u  = \Delta_u z_k^u,
\end{gathered}
\end{equation}
where $w^s \rightarrow z^s$ is the scheduling channel and $w^u \rightarrow z^u$ is the uncertainty channel.
After the exploration phase, the control input can now be defined as
\begin{equation}\label{eq:input}
u_k=K x_k +K_s w_k^s.
\end{equation}
The goal is to design $K$ and $K_s$ such that the closed-loop system is stable and the specified performance criterion is met.
The robust gain-scheduling configuration is illustrated in Figure~\ref{fig:GenPlant}.
\begin{figure}[hbtp]
\begin{center}
\includegraphics[width=0.2\textwidth]{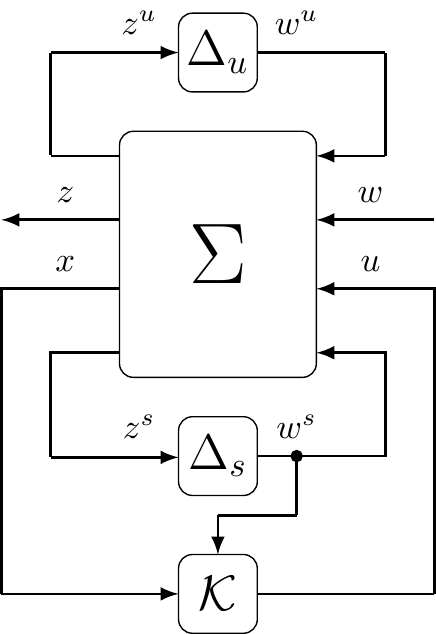}
\end{center}
\caption{Generalized plant view of the robust gain-scheduling problem. }
\label{fig:GenPlant}
\end{figure}
As can be seen in Figure~\ref{fig:GenPlant}, the open-loop system has two uncertainty channels, affected by $\Delta_u$ and $\Delta_s$.
The latter uncertainty $\Delta_s$ is taken into account for the controller via $w_k^s$, compare~\eqref{eq:input}, and hence plays the role of a scheduling variable.
This accounts for changes in the mean of the system parameters through data gathered in the exploration phase that is available after exploration at time $T$, and thereby learning the system dynamics better. The closed-loop system can be written as

\begin{align}\nonumber
&\begin{bmatrix}
x_{k+1}\\z_k^s \\ z_k^u \\z_k
\end{bmatrix}  = 
\begin{bmatrix}
\hat{A}_0 + \hat{B}_0 K & I +\hat{B}_0 K_s & I & I\\
\begin{bmatrix}
I\\K
\end{bmatrix} & \begin{bmatrix}
0 \\ K_s
\end{bmatrix} & 0 & 0 \\
\begin{bmatrix}
I\\K
\end{bmatrix} & \begin{bmatrix}
0 \\ K_s
\end{bmatrix} & 0 & 0 \\ 
C+DK & DK_s & 0 & D_w
\end{bmatrix}
\begin{bmatrix}
x_k\\w_k^s\\w_k^u\\w_k
\end{bmatrix},\\\label{closedloop}
&\qquad\qquad\qquad\qquad\quad w_k^s = \Delta_s z_k^s,\\\nonumber
&\qquad\qquad\qquad\qquad\quad w_k^u = \Delta_u z_k^u.
\end{align}

Given this formulation and suitable bounds on the blocks $\Delta_s$ and $\Delta_u$, we can use established methods from robust control and gain-scheduling to guarantee a desired performance specification. 
In particular, we consider the case where a desired quadratic performance specification is given on the performance channel $w_k\rightarrow z_k$, i.e. for initial condition $x=0$ for all signals $w\in\mathcal{L}_2$ with output $z$ of the closed loop, the following inequality should hold with some $\epsilon>0$:

\begin{align}
\label{eq:quad_perf}
\sum_{k=0}^{\infty}\begin{pmatrix}w_k\\z_k\end{pmatrix}^\top
\begin{pmatrix}Q_p&S_p\\S_p^\top&R_p\end{pmatrix}
\begin{pmatrix}w_k\\z_k\end{pmatrix}\leq -\epsilon \sum_{k=0}^{\infty}w_k^\top w_k,
\end{align}

where $R_p\succ 0$ is assumed. 
We note that standard design goals, such as a desired $\mathcal{L}_2$-gain of $\gamma$ are contained as a special case with $S_p=0$, $R_p=\frac{1}{\gamma}I$ and $Q_p=-\gamma I$ (c.f.~\cite[Prop.~3.12]{scherer2000linear}).
The following lemma provides a matrix inequality to design a robust gain scheduling controller satisfying such a performance specification, given suitable bounds on the blocks $\Delta_s,\Delta_u$. 
\begin{lemma}
\label{lemma:robust}
Suppose $\Delta_s\in\mathbf{\Delta}_s:=\{\Delta: \Delta^\top Q_s \Delta +R_s\succ 0\}$, $\Delta_u\in\mathbf{\Delta}_u:=\{\Delta: \Delta^\top Q_u \Delta  + R_u\succ 0\},$  with $R_u,R_S\succ 0$.
If there exists matrices $K_s,M,N$ and scalars $\lambda_s,\lambda_u>0$ satisfying  the matrix inequality~\eqref{eq:LMI_gain_schedule}, the closed loop~\eqref{closedloop} satisfies the quadratic performance bound~\eqref{eq:quad_perf} with $K=MN^{-1}$, i.e., $u_k=MN^{-1}x_k+K_sw_k^s$. 
\end{lemma}

\begin{proof}
The proof follows the arguments in~\cite{scherer2001lpv} for LPV control. 
First, note that the set definitions $\mathbf{\Delta}_s,\mathbf{\Delta}_u$ are linear in $Q_s,Q_u,R_s,R_u$ and thus remain valid if $(Q_s,R_s)$ and ($Q_u,R_u)$ are multiplied by some positive scalar $\lambda_s,\lambda_u>0$, respectively. Define $X=N^{-1}$ and $K=M N^{-1}$. The Schur complement of the LMI ~\eqref{eq:LMI_gain_schedule} is multiplied from left and right by $\text{diag}(N^{-1},I,I,I)$ to obtain

\begin{small}
\begin{equation}\label{eq:biginequality}
\begin{split}
\left[ \begin{array}{c}
*\\ *\\ \hline
*\\ *\\ \hline
*\\ *\\ \hline
*\\ *
\end{array}\right]^\top &
\left[
\begin{array}{c|c|c|c}
\begin{matrix} -X&0\\0&X \end{matrix} & 
\begin{matrix} 0&0\\0&0\end{matrix} &
\begin{matrix} 0&0\\0&0 \end{matrix} &
\begin{matrix} 0&0\\0&0 \end{matrix} \\
\hline
\begin{matrix} 0&0\\0&0 \end{matrix} &
\begin{matrix} \lambda_s P_s \end{matrix} & 
\begin{matrix} 0&0\\0&0 \end{matrix} &
\begin{matrix} 0&0\\0&0 \end{matrix} \\
\hline
\begin{matrix} 0&0\\0&0 \end{matrix} &
\begin{matrix} 0&0\\0&0 \end{matrix} &
\begin{matrix} \lambda_u P_u \end{matrix} & 
\begin{matrix} 0&0\\0&0 \end{matrix} \\
\hline
\begin{matrix} 0&0\\0&0 \end{matrix} &
\begin{matrix} 0&0\\0&0 \end{matrix} &
\begin{matrix} 0&0\\0&0 \end{matrix} &
\begin{matrix} Q_p&S_p\\S_p^\top&R_p \end{matrix}
\end{array}
\right] \\
\times & \left[ \begin{array}{cccc}
I&0&0&0 \\
\hat{A}_0+\hat{B}_0 K& I+ \hat{B}_0 K_s&I&I \\
\hline
0&I&0&0\\
\begin{bmatrix} I\\K \end{bmatrix} & \begin{bmatrix} 0\\K_s \end{bmatrix} & 0 & 0 \\
\hline
0&0&I&0 \\
 \begin{bmatrix} I\\K \end{bmatrix} & \begin{bmatrix} 0\\K_s \end{bmatrix} & 0 & 0\\
\hline
0&0&0&I\\
C+DK & DK_s & 0 & D_w\\
\end{array} \right] \prec 0,
\end{split}
\end{equation}
\end{small}
where $P_s=\begin{bmatrix}Q_s&0\\0&R_s\end{bmatrix}$ and $P_u=\begin{bmatrix}Q_u&0\\0&R_u
\end{bmatrix}$.

Using~\cite[Thm.~2]{scherer2001lpv}, quadratic performance is guaranteed if there exists a positive definite matrix $X=X^\top \succ 0$ satisfying the matrix inequality~\eqref{eq:biginequality}.\end{proof}

We note that for $\lambda_s,\lambda_u$ constant, inequality~\eqref{eq:LMI_gain_schedule} is an LMI and thus can be efficiently solved using line-search like techniques for $(\lambda_s,\lambda_u)\in\mathbb{R}^2$.

\begin{table*}
\begin{small}
\begin{align}
\label{eq:LMI_gain_schedule}
\begin{pmatrix}
\begin{array}{c|c}
\begin{bmatrix}
-N & 0 & 0 & (CN+DM)^\top S_p^\top\\
0 & \lambda_s Q_s & 0 & (DK_s)^\top S_p^\top \\
0 & 0 & \lambda_u Q_u & 0 \\
S_p(CN+DM) & S_p DK_s & 0 & Q_p +D_w^\top S_p^\top + S_pD_w
\end{bmatrix}&\star\\\hline
\begin{bmatrix}
\hat{A}_0N+\hat{B}_0 M& I+ \hat{B}_0 K_s&I&I \\
\begin{bmatrix} N\\M \end{bmatrix} & \begin{bmatrix} 0\\K_s \end{bmatrix} & 0 & 0 \\
 \begin{bmatrix} N\\M \end{bmatrix} & \begin{bmatrix} 0\\K_s \end{bmatrix} & 0 & 0\\
CN+DM & DK_s & 0 & D_w\\
\end{bmatrix}
&\begin{bmatrix}
-N & 0 & 0 & 0\\
0 & -\frac{1}{\lambda_s}R_s^{-1} & 0 & 0\\
0 & 0 & -\frac{1}{\lambda_u}R_u^{-1} & 0\\
0 & 0 & 0 & -R_p^{-1}
\end{bmatrix}
\end{array}
\end{pmatrix}
\prec 0.
\end{align}
\end{small}
\end{table*}

\subsection{Exploration and parameter estimation bounds}

In this paper, we consider a dual control objective where, during an initial exploration phase, uncertainty is reduced in order to design a robust controller based on the data collected during exploration.
The exploration controller is computed such that it excites the system sufficiently with a minimal \emph{robust} LQR cost, based on initial parameter estimates.
The exploration controller takes the form
\begin{equation}\label{eq:exploration_input}
u_k=K_ex_k + e_k,\quad k=0,\dots, T
\end{equation}
with a robustly stabilizing $K_e$ and noise $e_k \sim \mathcal{N}(0,\Sigma)$.
Based on initial estimates of the system dynamics and the associated uncertainty bound, $K_e$ and $\Sigma$ are computed such that they minimize a robust LQR cost $\sum_{k=0}^{\infty}x_k^\top Qx_k+u_k^\top Ru_k$.
Similar to~\cite{ferizbegovic2019learning}, this robust LQR cost can be computed as the $\mathcal{H}_2$-norm of the uncertain closed loop system $x_{k+1}=(\hat{A}_0+\hat{B}_0K_e)x_k+w_k$ and $y_k=\begin{bmatrix}
Q^{\frac{1}{2}}\\R^{\frac{1}{2}}K_e
\end{bmatrix}x_k$.
To be more precise, the robust LQR cost of the exploration controller is computed as
\begin{equation}\label{eq:robust_LQR}
\begin{split}
\underset{t_e,Z_e,Y_e,W_e}{\min} & \quad\text{tr } Y_e\\
\text{s.t.}& \quad S_1(W_e,Y_e,Z_e)\succeq 0,~t_e>0\\
& \quad S_e(t_e,Z_e,W_e,\Sigma,D_0,\hat{A}_0,\hat{B}_0) \succeq 0
\end{split}
\end{equation}
where $S_1,S_e$ are defined as
\begin{align*}
S_1(W_e,Y_e,Z_e)&=\left[ \begin{array}{c|c}
Y_e & \begin{matrix}
Q^{\frac{1}{2}}W_e\\R^{\frac{1}{2}} Z_e^\top
\end{matrix} \\
\hline
\begin{matrix}
W_eQ^{\frac{1}{2}} & Z_e R^{\frac{1}{2}}
\end{matrix} & W_e
\end{array} \right],\\
S_e(t_e,Z_e,W_e,\Sigma)&=\begin{bmatrix}
H_e & F_e & G_e\\
F_e^\top & C_e-t_eI & 0\\
G_e^\top & 0 & t_eD_0
\end{bmatrix}
\end{align*}
with 
\begin{align*}
    &H_e=\begin{bmatrix} W_e&0\\0&\Sigma \end{bmatrix},\>\>
    F_e=\begin{bmatrix}
    W_e\hat{A}_0^\top + Z_e\hat{B}_0^\top\\ \Sigma \hat{B}_0^\top
    \end{bmatrix},\\
    &G_e = \begin{bmatrix}
    -W_e&-Z_e\\0&-\Sigma
    \end{bmatrix},\>\>Z_e=W_eK_e^\top,\>\>C_e=W_e-\sigma_w^2I.
\end{align*}
A more detailed derivation and additional explanations are provided in~\cite{ferizbegovic2019learning}.
In~\eqref{eq:robust_LQR}, $W_e$ denotes the controllability Gramian, which plays an essential role to propagate the influence of the exploration phase on the parameters estimates based on the new data $\mathcal{D}_T$.
Similar to~\cite{ferizbegovic2019learning}, we make the following assumption.
\begin{assumption}\label{a2} 
For the system (\ref{sys}) evolving under an exploration controller~\eqref{eq:exploration_input} with $K_e,W_e$ satisfying the constraints in~\eqref{eq:robust_LQR}, the empirical covariance can be approximated via a solution $W_e$ of~\eqref{eq:robust_LQR} as
\begin{equation}\label{eq:a2}
\sum_{t=0}^{T-1}\begin{bmatrix}
x_t\\ u_t
\end{bmatrix}
\begin{bmatrix}
x_t\\ u_t
\end{bmatrix}^\top \approx T\begin{bmatrix}
W_e&W_eK_e^\top \\
K_e W_e & K_e W_e K_e^\top + \Sigma
\end{bmatrix}.
\end{equation}
\end{assumption}
Assumption~\ref{a2} implies that the empirical covariance can be approximated via its stationary distribution $\Sigma_{xx}=\mathbb{E}[xx^\top]$, which is in turn approximated by the worst-case state covariance, i.e., by $W_e$ satisfying~\eqref{eq:robust_LQR}.
Clearly, this is an assumption that is not guaranteed to hold, but it is usually a good approximation.
Based on~\eqref{eq:a2}, the uncertainty bound $D_T$ can be computed as
\begin{equation}
\label{eq:explore_2}
D_T=D_0+\frac{T}{\sigma_w^2 c_\delta} \begin{bmatrix}
W_e& Z_e\\Z_e^\top & Z_e^\top W_e^{-1} Z_e + \Sigma
\end{bmatrix}.
\end{equation}
Since the uncertainty bound $D_T$ is a nonlinear function of $Z_e$ and $W_e$, we compute an affine lower bound of it as in \cite[Lemma 1]{ferizbegovic2019learning}
\begin{equation}
\label{eq:explore_3}
\begin{bmatrix}
W_e & Z_e \\ Z_e^\top & Z_e^\top W_e^{-1}Z_e
\end{bmatrix} \succeq
\begin{bmatrix} W_e \\ Z_e^\top \end{bmatrix} V + V^\top \begin{bmatrix} W_e \\ Z_e^\top \end{bmatrix}^\top - V^\top W_e V.
\end{equation}
For a fixed $V$, this leads to an affine lower bound on $D_T$, denoted as $\overline{D}_T$, via~\eqref{eq:explore_2}. The bound is tight when $\begin{bmatrix}W_e&Z_e \end{bmatrix}=W_e V$, so it is optimal to choose $V=W_e^{-1}	\begin{bmatrix}W_e&Z_e \end{bmatrix}=\begin{bmatrix}I&K_e^\top \end{bmatrix}$.
However, $K_e$ is not known at this point, and hence a candidate $K_0$ is used instead to compute $V$, which can be computed, e.g., based on a robust LQR for the nominal model~\cite{ferizbegovic2019learning}.

An essential ingredient of the proposed approach is the handling of the uncertainty bounds $D_0,D_T$, which influence the uncertain parameters $\Delta_u,\Delta_s$ in~\eqref{closedloop}, as derived in the following proposition.
\begin{proposition}\label{p1}
Let Assumption \ref{a1} hold, where $\hat{A}_0$ and $\hat{B}_0$ are the initial estimates. Let
\begin{equation*}
\Delta_0=\begin{bmatrix}
\hat{A}_0 - A_{tr} & \hat{B}_0 - B_{tr}
\end{bmatrix}.
\end{equation*}
Then, with probability $1-\delta$, we have 
\begin{equation}\label{delta0bound}
\Delta_0 \in \mathbf{\Delta}_0:=\{\Delta_0: \Delta_0^\top \Delta_0 \preceq D_0^{-1}\},
\end{equation}
and with probability $1-\delta$, we have
\begin{equation}\label{ubound}
\Delta_u \in \mathbf{\Delta}_u :=\{\Delta_u:\Delta_u^\top \Delta_u \preceq D_T^{-1}\}.
\end{equation}
If (\ref{delta0bound}) and (\ref{ubound}) hold, then for any $\epsilon>0$ we have
\begin{equation}
\label{deltasbound}
\begin{split}
\Delta_s \in \mathbf{\Delta}_s =\Bigg\{\Delta_s:\Delta_s^\top \Delta_s \preceq & \left( 1+ \frac{1}{\epsilon}\right) D_0^{-1} \\ & + \left( 1+\epsilon \right) D_T^{-1}\Bigg\}.
\end{split}
\end{equation}
\end{proposition}

\begin{proof}
Let
\begin{equation}
\Delta_0=\begin{bmatrix}
\hat{A}_0 - A_{tr} & \hat{B}_0 - B_{tr}
\end{bmatrix}.
\end{equation}
By Lemma \ref{lem1}, with $0 < \delta <1$, the following hold with probability $1-\delta$.
\begin{equation}
\begin{split}
\Delta_0^\top \Delta_0 & \preceq D_0^{-1},\\
\Delta_u^\top \Delta_u & \preceq D_T^{-1}.
\end{split}
\end{equation}

The scheduling block can now be represented as
\begin{equation}
\Delta_s=-(\Delta_0+\Delta_u).
\end{equation}

To derive a probabilistic bound for $\Delta_s$, we have
\begin{equation}
\begin{split}
\Delta_s^\top \Delta_s & = (\Delta_0+\Delta_u)^\top (\Delta_0+\Delta_u)\\
&=\Delta_0^\top\Delta_0 +\Delta_0^\top\Delta_u +\Delta_u^\top\Delta_0 + \Delta_u^\top\Delta_u\\
&\preceq \left( 1+ \frac{1}{\epsilon}\right) \Delta_0^\top\Delta_0 + \left( 1+\epsilon \right)\Delta_u^\top\Delta_u.
\end{split}
\end{equation}
The third inequality follows by Young's inequality which implies that for every $\epsilon > 0$, $\Delta_0^\top\Delta_u +\Delta_u^\top\Delta_0 \leq \left(\frac{1}{\epsilon}\right) \Delta_0^\top\Delta_0 + \epsilon \Delta_u^\top\Delta_u $.
Therefore, the bound for $\Delta_s$ is
\begin{align*}
\Delta_s^\top \Delta_s \preceq \left( 1+ \frac{1}{\epsilon}\right) D_0^{-1} + \left( 1+\epsilon \right) D_T^{-1}.&\qedhere
\end{align*}
\end{proof}
The relation between the different sets is visualized in Figure~\ref{fig:Delta}.
Using Lemma~\ref{lem1} with the initial data $\mathcal{D}_0$, we know that the true system parameters $\theta_{tr}$ are in some ellipse $\mathbf{\Delta}_0$ around the initial parameter estimate $\hat{\theta}_0$. Using Lemma~\ref{lem1} after the exploration, we know that the true parameter $\theta_{tr}$ is contained in an ellipse $\mathbf{\Delta}_u$ around the new point estimate $\hat{\theta}_T$. Combining both of these bounds we know that the new point estimate $\hat{\theta}_T$, and thus the scheduling variable $\Delta_s$, is contained in an ellipse $\mathbf{\Delta}_s$ around the initial point estimate.
\begin{figure}[hbtp]
\begin{center}
\includegraphics[width=0.3\textwidth]{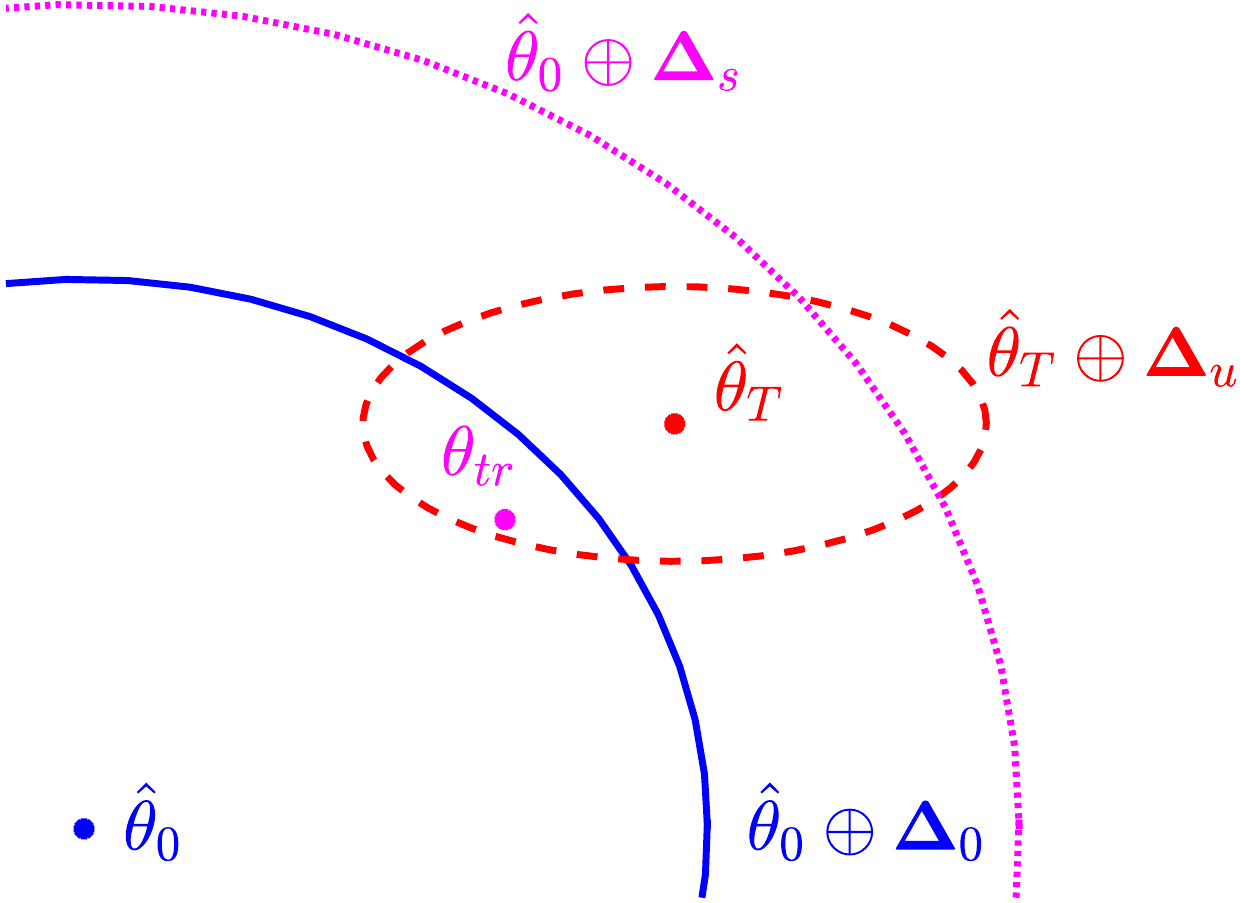}
\end{center}
\caption{Illustration of the sets $\mathbf{\Delta}_0,\mathbf{\Delta}_s,\mathbf{\Delta}_u$ (shifted w.r.t $\hat{\theta}_0,\hat{\theta}_T$) with the true parameters $\theta_{tr}$, the initial parameter estimate $\hat{\theta}_0$ and the estimate resulting from the exploration $\hat{\theta}_T$. }
\label{fig:Delta}
\end{figure}


%% file: dualcontrol.tex
\section{Dual Control}
\label{sec:dual}
In this section, the proposed robust dual control strategy is presented. 
The overall proposed algorithm is summarized in Algorithm~\ref{alg:main} in Section~\ref{sec:dual_1} and its theoretical properties are analyzed in Theorem~\ref{thm:main} in Section~\ref{sec:dual_2}. 
Finally, Section~\ref{sec:dual_3} discusses the approach relative to existing methods.
\subsection{Proposed Algorithm}
\label{sec:dual_1}
The overall goal of the proposed dual approach is to design a structured exploration strategy, such that the designed controller satisfies some desired quadratic performance~\eqref{eq:quad_perf} with high probability $1-\delta\in(0,1)$.
Since we assume that we do not have a prior on the model, we start with some random (ideally persistently exciting) exploration over $N_0\in\mathbb{N}$ steps to obtain initial data $\mathcal{D}_0$ (Assumption ~\ref{a1}). 
Then we use the least mean squares estimate~\eqref{eq:LMS} to obtain initial estimates $\hat{A}_0,\hat{B}_0$ and a high probability uncertainty bound $D_0^{-1}$ (c.f. Lemma~\ref{lem1} and~\eqref{delta0bound} in Prop.~\ref{p1}). 
Before solving the dual control problem, we first seek some initial candidate feedback $K_0\in\mathbb{R}^{m\times n}$. 
This step is in principle not necessary, however, in case $K_0-K_e$ is small, this step can greatly reduce the conservatism in the convex relaxation~\eqref{eq:explore_3}. 
Thus, as also suggested in~\cite[Sec.~IV.C]{ferizbegovic2019learning}, we compute $K_0$ as a robust LQR for the nominal model~\eqref{eq:robust_LQR}. 

In order to emphasize its dependence on the variables and uncertainty parameters, we denote the matrix in~\eqref{eq:LMI_gain_schedule} with $Q_s=-I,Q_u=-I$ by
\begin{align*}
S_2(K_s,M,N,\lambda_s,\lambda_u,R_s^{-1},R_{u}^{-1}).
\end{align*}

Furthermore, the satisfaction of~\eqref{eq:LMI_gain_schedule} is ensured if 
$\Delta_s^\top \Delta_s \prec D_s^{-1}$, $\Delta_u^\top \Delta_u\prec D_T^{-1}$, where $D_s$ denotes the uncertainty bound associated with $\Delta_s$. Since the set inclusions~\eqref{delta0bound}--\eqref{deltasbound} from Proposition~\ref{p1} hold, it suffices to show 
\begin{align}
\label{eq:D_s_intermediate_1}
D_s^{-1}\succ \frac{1+\epsilon}{\epsilon}D_0^{-1}+(1+\epsilon)D_T^{-1},   
\end{align}
By applying the Woodbury matrix identity to~\eqref{eq:D_s_intermediate_1} and multiplying by (1+$\epsilon$), we get
\begin{align}
\label{eq:D_s_intermediate2}
    (1+\epsilon){D}_s \succ \epsilon D_0 - \epsilon D_0 \left( D_T + \epsilon D_0 \right)^{-1} \epsilon D_0. 
\end{align}
Applying the Schur complement to~\eqref{eq:D_s_intermediate2} results in the following equivalent LMI
\begin{equation*}
S_3(\epsilon,D_0,D_T,D_s)=
    \begin{bmatrix}
      \epsilon D_0-(1+\epsilon)D_s & \epsilon D_0\\
     \epsilon D_0 &  D_T + \epsilon D_0 
    \end{bmatrix}.
\end{equation*}

Given $K_0$, $D_0$, $\hat{A}_0$, $\hat{B}_0$, $\delta$, $Q_p$, $S_p$, $R_p$, $T$ and some fixed $\epsilon$, $t_e$, $\lambda_s$, $\lambda_u>0$, we solve the following semi-definite program (SDP), which is a combination of the robust gain-scheduling problem~\eqref{eq:LMI_gain_schedule} and the exploration inequalities~\eqref{eq:robust_LQR}, \eqref{eq:explore_2}:
\begin{subequations}
\label{eq:opt_problem}
\begin{align}
&\underset{\begin{subarray}
\{W_e, Z_e,Y_e,\Sigma\\
{K_s, M, N, \overline{D}_T,D_s}
\end{subarray}}{\inf}  \quad\text{tr}(Y_e) \\\label{eq:opt_problem1}
\text{s.t. }
& \quad S_1(W_e,Y_e,Z_e)\succeq 0\\\label{eq:opt_problem3}
& \quad S_e(t_e,Z_e,W_e,\Sigma,D_0,\hat{A}_0,\hat{B}_0) \succeq 0\\
\label{S3}& \quad S_2(K_s,M,N,\lambda_s, \lambda_u, D_s,\overline{D}_T) \prec 0\\
\label{eq:opt_problem_D_s}
& \quad S_3(\epsilon,D_0,\overline{D}_T,D_s) \succ 0\\\nonumber
& \quad \frac{T}{\sigma_w^2 c_\delta}\begin{bmatrix}
 W_e & Z_e\\Z_e^\top & Z_e^\top K_0^\top + K_0 Z_e - K_0 W_e K_0^\top + \Sigma\end{bmatrix}\\
\label{eq:opt_problem_D_T}
 &\qquad\quad+D_0-\overline{D}_T\succ0.
\end{align}
\end{subequations}

Solving this optimization problem directly leads to the controller parameters required for the implementation, i.e., the exploration controller $K_e=Z_e^\top W_e^{-1}$, the exploration variance $\Sigma$, and the robust gain scheduled controller parameters $K_s$ and $K=MN^{-1}$. 
Essentially,~\eqref{eq:opt_problem1}-\eqref{eq:opt_problem3} are needed to compute the cost $\text{tr}(Y_e)$ of the controller during the exploration phase, compare~\eqref{eq:robust_LQR}.
Moreover,~\eqref{S3} contains the main robust control LMI (compare Lemma~\ref{lemma:robust}) which returns a common Lyapunov function $N\succ0$ as well as controller parameters $M,K_s$ which guarantee robust performance of the closed loop~\eqref{closedloop} for all uncertainties $\Delta_u,\Delta_s$ satisfying $\Delta_u^\top\Delta_u\prec D_T^{-1}$ and $\Delta_s^\top\Delta_s\prec D_s^{-1}$.
In this context, (a bound on) the data obtained during exploration is approximated via $\overline{D}_T-D_0$, which implies that the uncertainties for robust controller design, i.e., the values $\overline{D}_T$ and $D_s$, in turn depend on the controller during the exploration phase $K_e$ through~\eqref{eq:opt_problem_D_s} and~\eqref{eq:opt_problem_D_T}.
This couples the exporation and robust control, thus resulting in a \emph{dual effect} of the proposed controller.

Regarding the computational complexity of~\eqref{eq:opt_problem}, we note that for $\epsilon$, $t_e$, $\lambda_s$, $\lambda_u>0$ fixed, this is a standard (small-scale) semi-definite program (SDP), which can be efficiently solved. 
Hence, the optimization problem can be solved by using a line-search like procedure (or gridding) for the variables $\epsilon$, $t_e$, $\lambda_s$, $\lambda_u>0$ and solving the SDP in an inner loop.


After solving~\eqref{eq:opt_problem}, we apply the targeted exploration sequence $u_t=K_ex_t+e_t$, $e_t\sim\mathcal{N}(0,\Sigma)$ for $t=0,\dots,T$.
Next, with the new data $\mathcal{D}_0\cup\mathcal{D}_T$, we use the least mean square estimate~\eqref{eq:LMS} to obtain an improved/updated estimate $\hat{A}_T$, $\hat{B}_T$ and a new bound $D_T^{-1}$ on the uncertainty. 
Then, we can directly apply the designed gain-scheduling controller with the new scheduling variable $\Delta_s=\begin{pmatrix}\hat{A}_T-\hat{A}_0&\hat{B}_T-\hat{B}_0\end{pmatrix}$. Using~\eqref{closedloop}, this controller can be explicitly written as a state feedback control law $K_{\text{new}}$ using
\begin{align}
\label{eq:K_new}
u_k&=Kx_k+K_sw_k^s\\
&=Kx_k+K_s((\hat{A}_T-\hat{A}_0)x_k+(\hat{B}_T-\hat{B}_0)u_k)\nonumber\\
&=(I_m-K_s(\hat{B}_T-\hat{B}_0))^{-1}(K+K_s(\hat{A}_T-\hat{A}_0))x_k\nonumber\\\nonumber
&=:K_{\text{new}}x_k.
\end{align}
We note that $(I-K_s(\hat{B}_T-\hat{B}_0))$ is non-singular (with high probability) due to the equivalence in~\cite[Thm.~2]{scherer2001lpv}. The overall procedure is summarized in Algorithm~\ref{alg:main}. 
\begin{algorithm}[H]
\caption{Dual control using gain-scheduling}
\label{alg:main}
\begin{algorithmic}[1]
\State Specify confidence level $\delta\in(0,1)$, quadratic performance ($Q_p,\,S_p,\,R_p$)~\eqref{eq:quad_perf}, exploration cost $Q,\,R\succ 0$, initial and targeted exploration length $N_0,\,T$.
\State Random exploration to obtain initial data $\mathcal{D}_0$ (Ass.~\ref{a1}).
\Statex $\Rightarrow$ Initial estimates $\hat{A}_0,\,\hat{B}_0$ and uncertainty bound $D_0^{-1}$,
\Statex compute robust LQR controller $K_0$~\eqref{eq:robust_LQR}.
\State Solve the optimization problem~\eqref{eq:opt_problem} for different values $\epsilon,\,t_e,\,\lambda_s,\,\lambda_u>0$ (e.g., via line-search in an outer loop).
\Statex $\Rightarrow$ Exploration sequence $K_e=Z^\top W^{-1}$, $\Sigma$ and gain-scheduled controller $K_s$, $K=MN^{-1}$.
\State Apply the exploration input $u_k=K_ex_k+e_k$, $e_k\sim\mathcal{N}(0,\Sigma)$ for $k=0,\dots, T$.
\State Update estimates $\hat{A}_T,\,\hat{B}_T$ using new data.
\State Compute the equivalent state-feedback $K_{new}$ and apply the feedback $u_k=K_{\text{new}}x_k$, $k>T$.
\end{algorithmic}
\end{algorithm}

\subsection{Theoretical analysis}
\label{sec:dual_2}

The following result proves that Algorithm~\ref{alg:main} leads to a controller with closed-loop guarantees.

\begin{theorem}
\label{thm:main}
Let Assumptions~\ref{a1}--\ref{a2} hold, suppose~\eqref{eq:opt_problem} is feasible and Algorithm~\ref{alg:main} is applied. 
Assume further that the set inclusions~\eqref{delta0bound}--\eqref{deltasbound} from Proposition~\ref{p1} hold.
Then the state-feedback $K_{\text{new}}$ from~\eqref{eq:K_new} is well-defined and satisfies the quadratic performance bound~\eqref{eq:quad_perf}. 
\end{theorem}
\begin{proof}
First, we recap that Lemma~\ref{lemma:robust} guarantees the performance bound~\eqref{eq:quad_perf}, assuming suitable bounds on $\Delta_s,\Delta_u$. 
Then, we show that exploration inequalities in combination with Assumption~\ref{a2} ensure the bounds on $\Delta_s,\Delta_u$.\\
\textbf{Part I. }
According to Lemma~\ref{lemma:robust}, satisfaction of the matrix inequality~\eqref{S3} guarantees that the robust gain-scheduling controller $u=MN^{-1}x_k+K_sw_k$ ensures the quadratic performance bound~\eqref{eq:quad_perf}, if $\Delta_s^\top \Delta_s \prec D_s^{-1}$, $\Delta_u^\top \Delta_u\prec \overline{D}_T^{-1}$.
Moreover, it is a direct consequence of the synthesis LMI that $K_{\text{new}}$ is well-posed.
Thus, it only remains to show that $\Delta_s^\top \Delta_s \prec D_s^{-1}$, $\Delta_u^\top \Delta_u\prec \overline{D}_T^{-1}$.\\
\textbf{Part II. }
Since the set inclusions~\eqref{delta0bound}--\eqref{deltasbound} from Proposition~\ref{p1} hold, it suffices to show $\overline{D}_T^{-1}\succ D_T^{-1}$ and 
\begin{align}
\label{eq:D_s_intermediate}
D_s^{-1}\succ \frac{1+\epsilon}{\epsilon}D_0^{-1}+(1+\epsilon)\overline{D}_T^{-1},   
\end{align}
Assumption~\ref{a2} ensures that the bound~\eqref{eq:explore_2} holds. 
The convex relaxation~\eqref{eq:explore_3} (c.f.~\cite[Lemma 1]{ferizbegovic2019learning})  in combination with inequality~\eqref{eq:opt_problem_D_T} ensures that $D_T\succ \overline{D}_T$ and thus $\Delta_u^\top \Delta_u \prec \overline{D}_T^{-1}$. 
Finally, as shown earlier, inequality~\eqref{eq:opt_problem_D_s} is equivalent to~\eqref{eq:D_s_intermediate}, which implies $\Delta_s^\top \Delta_s\prec {D}^{-1}_s$.
\end{proof}

We point out that, the since the properties in Proposition~\ref{p1} only hold with some probability $1-\delta$, the quadratic performance~\eqref{eq:quad_perf} only holds with some probability, which is inherent in the considered stochastic/Gaussian setup.

\subsection{Discussion}
\label{sec:dual_3}
The proposed method detailed in Algorithm~\ref{alg:main} combines structured exploration techniques as developed in~\cite{umenberger2019robust,ferizbegovic2019learning} and robust gain scheduling controller design.
Given an initial data set (compare Assumption~\ref{a1}) and a quadratic performance specification $Q_p,S_p,R_p$ on the channel $w\mapsto z$, Theorem~\ref{thm:main} implies that Algorithm~\ref{alg:main} guarantees robust performance for the closed loop with input $u_k=K_{\text{new}}x_k$, after an initial exploration phase whose worst-case cost is minimized simultaneously with the controller design.
The influence of the exploring controller $u_k=K_ex_k+e_k$, $e\sim\mathcal{N}(0,\Sigma)$, on the performance after exploration is quantified by (approximately) predicting the future uncertainty depending on $K_e$ and $\Sigma$ via \eqref{eq:explore_2}--\eqref{eq:explore_3}. 

Compared to previous works~\cite{umenberger2019robust, ferizbegovic2019learning, barenthin2008identification, iannelli2019structured}, the key difference of the present approach is that the mean of the parameter estimates after exploration is taken into account by considering it as a scheduling variable via $w_k^s=\Delta_sz_k^s$.
Initially, it is only known that $\Delta_s\in\mathbf{\Delta}_s$ (compare Lemma~\ref{lemma:robust}), but after exploration $\Delta_s$ is available and can hence be exploited for controller design.
This is in contrast to existing works, which simply assumed $\Delta_s=0$, i.e., the mean value of the parameter estimates does not change over time.
An important observation is that, according to~\eqref{eq:K_new}, the state-feedback $K_{\text{new}}$ depends on $\hat{A}_T,\hat{B}_T$ and hence, on the data $\mathcal{D}_T$ obtained during time steps $0$ through $T$.
This means that the proposed controller explicitly exploits measurements during the exploration phase, which was not the case in~\cite{umenberger2019robust, ferizbegovic2019learning, barenthin2008identification, iannelli2019structured}. Furthermore, \cite{umenberger2019robust, ferizbegovic2019learning, barenthin2008identification, iannelli2019structured} require a repeated LMI based design after the exploration phase, which is not the case in our formulation wherein we pre-compute a closed-form solution that guarantees quadratic performance based on a predicted bound of the exploration data.

Theorem~\ref{thm:main} requires that Assumption~\ref{a1} holds, which is a non-restrictive condition on the initial data and parameters.
On the contrary, Assumption~\ref{a2} is essentially an approximation on the empirical covariance, which is required to predict the influence of the exploration phase on the parameter estimates.
While Assumption~\ref{a2} is generally not guaranteed to hold, it is approximately satisfied in practice and its validity can be verified a posteriori, i.e., after the exploration phase.

We briefly wish to elaborate on the impact of different values $c_\delta,\sigma_w$ corresponding to different noise and confidence levels. 
Assuming a fixed initial uncertainty $D_0$ is given, $c_\delta,\sigma_w$ have the same effect and only appear in~\eqref{eq:opt_problem_D_T} to determine $\overline{D}_T$. 
In case we increase $c_\delta$ and/or $\sigma_w^2$ (assuming $D_0$ is fixed), the optimal controller parameters $K,K_s$ resulting from~\eqref{eq:opt_problem} remain unchanged and only the cost of the exploration ($Y_e,W_e,Z_e,\Sigma$) increases proportionally. 
This is natural, as a higher noise level and/or a higher desired confidence level requires a stronger excitation to yield the required model quality. 
Thus, since the magnitude of noise and/or confidence level does not directly impact the resulting controller $K,K_s$ (although $K_{\text{new}}$ may change), the main structural property  that would \emph{qualitatively} change the shape of the resulting robust dual control strategy would be varying noise levels for the different states.

%% file: conclusion.tex
 \section{Conclusion}\label{sec:conclusion}
In this paper, we formulate a novel dual control approach for linear time-invariant systems with performance guarantees based on gain-scheduling.
We propose an LMI-based controller design procedure which simultaneously computes a controller to apply during an exploration phase as well as a robust controller for closed-loop performance after exploration.
Similar to~\cite{ferizbegovic2019learning}, the influence of the exploration on the closed loop is quantified by predicting the future uncertainty of the system parameters.
The key difference is that we account for the change in the mean estimate of system parameters after exploration by formulating an LPV closed-loop system and selecting the uncertain system parameters as a scheduling variable. 
In contrast to existing methods, the robust controller takes the estimates after exploration into account and therefore, it depends explicitly on the data obtained during exploration.
Finally, we prove desirable theoretical properties of the proposed approach.
An interesting issue for future research is a detailed comparison of the presented dual controller to existing alternatives.